\documentclass{llncs}
\usepackage{makeidx}  
\usepackage{paralist}
\usepackage{graphicx}
\usepackage{epstopdf}
\usepackage[linesnumbered,ruled]{algorithm2e}
\usepackage{amsmath}
\usepackage{tabularx}
\usepackage{url}
\begin{document}
\frontmatter          
\pagestyle{headings}  
%
%
\mainmatter              
\title{Minimum Number of Test Paths for Prime Path and other Structural Coverage Criteria \thanks{Author's final version. The original publication is available at {https://link.springer.com/chapter/10.1007/978-3-662-44857-1\_5}}}
%
%
\author{Anurag Dwarakanath \and Aruna Jankiti}
\authorrunning{Anurag Dwarakanath et al.} 
%
%
\institute{Accenture Technology Labs, Bangalore, India\\
\email{anurag.dwarakanath@accenture.com, jankiti.aruna@accenture.com}
}

\maketitle              

\begin{abstract}
The software system under test can be modeled as a graph comprising of a set of vertices, \(V\) and a set of edges, \(E\). Test Cases are Test Paths over the graph meeting a particular test criterion. In this paper, we present a method to achieve the minimum number of Test Paths needed to cover different structural coverage criteria. Our method can accommodate Prime Path, Edge-Pair, Simple \& Complete Round Trip, Edge and Node coverage criteria. Our method obtains the optimal solution by transforming the graph into a flow graph and solving the minimum flow problem. We present an algorithm for the minimum flow problem that matches the best known solution complexity of
\begin{math} O\left(\left\vert{V}\right\vert\left\vert{E}\right\vert\right)\end{math}. Our method is evaluated through two sets of tests. In the first, we test against graphs representing actual software. In the second test, we create random graphs of varying complexity. In each test we measure the number of Test Paths, the length of Test Paths, the lower bound on minimum number of Test Paths and the execution time. 

\end{abstract}
\section{Introduction \label{section:1}}
In model based testing, the software artifact describing the system under test (SUT) is abstracted through a model \cite{Ammann_Offutt_2008}. The model can be created from requirements, design or code \cite{Ammann_Offutt_2008}. Models are typically represented in the form of a graph, comprising of vertices and edges. Two special vertices are marked, the source vertex, \(s\), and the sink vertex, \(t\). The graph is then used to generate Test Cases. A Test Case consists of a Test Path, Test Data, and Expected Results. The Test Path is a path in the graph from \(s\) to \(t\). The number of Test Paths needed to test the SUT is determined by the coverage criterion. For example, Node Coverage implies that the set of Test Paths should collectively visit all vertices in the graph. Similarly, Edge Coverage requires the Test Paths to visit all edges in the graph and Prime Path coverage requires the Test Paths to tour a particular set of paths. Prime Path coverage provides a better quality of coverage as it subsumes Node Coverage and Edge Coverage \cite{Ammann_Offutt_2008}. 

Given a particular coverage criterion, it is important to determine a small set of Test Paths which satisfy the criterion. The number and the overall length of Test Paths directly impacts the amount of time needed to manually execute the test cases. Finding the minimum number of Test Paths needed to satisfy a particular  criterion is non-trivial even for the simplest case of node coverage. Using better coverage criterion like Prime Path complicates the problem further since even a small graph can have a large number of Prime Paths \cite{Kaminski_Praphamontripong_Ammann_Offutt_2010}. 

In this paper, we present a method to obtain the optimal solution for the minimum number of Test Paths for different structural coverage criteria. 
The contribution of our work includes \begin{inparaenum} [\itshape a\upshape)] \item a generic method for obtaining the minimal number of Test Paths for different structural coverage criteria; \item identification of a lower bound on the minimum number which the experimental results show to perform very well; \item a new method for computation of minimum number of Test Paths which matches the best known solution complexity.\end{inparaenum}

The identification of the minimum number of Test Paths, although possible in polynomial time, requires a number of graph transformations. We thus test the applicability of our solution against a simple algorithm for Test Path generation.  The tests compare the number and length of Test Paths identified and the algorithm execution time. These metrics are compared on a set of manually created graphs representing actual code and a set of randomly generated graphs. 

The paper is structured as follows. We present the related work in Section \ref{section:2}. The method of generating the minimum Test Paths is presented in Section \ref{section:3}. We present the experimental results in Section \ref{section:4} and conclude in Section \ref{section:5}.

\section{Related Work \label{section:2}}
The input for the identification of minimum Test Paths is a directed graph \(G_{1}=\left(V_{1},E_{1}\right)\), where \(V_{1}\) represents the vertex set and \(E_{1}\) represents the edge set. \(G_{1}\) is typically a control flow graph of the SUT. \(G_{1}\) may contain cycles.  \(V_{1}\) contains the source vertex, \(s\), and the sink vertex, \(t\). \(s\) conceptually depicts the point where the execution of a Test Case begins and \(t\) denotes the point where the execution ends. The in-degree of \(s\) is \(0\), and the out-degree of \(t\) is \(0\). We focus on the case where \(\left\vert{s}\right\vert=\left\vert{t}\right\vert=1\). There is no loss in generality since a graph with more than one \(s\) or \(t\) can be converted to a graph with a single \(s\) and a single \(t\), by adding two new vertices, \(s^{'}\), \(t^{'}\) with edges between \(s^{'}\) and every source, and edges between every sink and \(t^{'}\).

A path is a sequence of vertices \(v_{0},\ldots,v_{r}\) with a sequence of edges \(e_{0},\ldots,\allowbreak e_{\left(r-1\right)}\), where \(e_{i}=\left(v_{i} \allowbreak,v_{\left(i+1\right)}\right)\), \(v_{i} \in V_{1}\), \(e_{i} \in E_{1} \forall i\). An \(s-t\) path is a path where \(v_{0}=s\) and \(v_{r}=t\). All Test Paths are thus \(s-t\) paths. A path, \(p\) is said to visit vertex \(k\) (or edge \(e\)) if \(k\) (or \(e\)) is in \(p\).  A path \(p\) is said to tour (or cover) a path \(q\), if \(q\) is a sub-path of \(p\). A path is simple if no vertex is visited more than once with the exception that the first and last nodes may be identical \cite{Ammann_Offutt_2008}. A Prime Path is a simple path that is not a sub-path of any other simple path \cite{Ammann_Offutt_2008}. 

A test requirement is a specific aspect that a Test Path satisfies. \(TR\) denotes the set of test requirements. For example, for Prime Path coverage, \(TR\) = set of Prime Paths. For Edge Coverage, \(TR\) = the set of edges, etc.

A simple solution approach for the generation of Test Paths that cover all Prime Paths has been presented in \cite{Ammann_Offutt_2008}. The solution `extends' every Prime Path to visit \(s\) and \(t\) thus forming a Test Path. The algorithm does not attempt to minimize the number of Test Paths but is extremely fast in execution.

The problem of minimizing Test Paths for Prime Path coverage has been recently studied in \cite{Li_Li_Offutt_2012}. The authors formulate the problem as a variant of the shortest superstring problem, which is NP-complete \cite{Li_Li_Offutt_2012}. The authors then use known approximation algorithms to solve the problem. In our work, we formulate the problem such that the optimal solution is obtained in polynomial time.

While the work in \cite{Li_Li_Offutt_2012} specifically deals with Prime Paths, the problem of minimum Test Path to cover a given set of paths (not necessarily Prime Paths) has been studied earlier in \cite{Ntafos_Hakimi_1979}. Here, the given graph, \(G_{1}\), is transformed into a flow network such that the minimum flow gives the minimum number of Test Paths. The authors use a generic algorithm for minimum flow. One such simple algorithm for minimum flow is the Decreasing Path algorithm \cite{Ciurea_Ciupala_2004} which leads to the solution in \(O\left(\left\vert{V}\right\vert^2\left\vert{E}\right\vert\right)\). However, the methodology in \cite{Ntafos_Hakimi_1979} has an inaccuracy which can lead to incorrect results. Consider the graphs in Fig. \ref{fig:1}. The Prime Paths of Fig. \ref{fig:1}(a) are \(\{1,2,3\}\) and \(\{2,2\}\). The technique in \cite{Ntafos_Hakimi_1979} makes the graph acyclic by introducing a new vertex \(2^{'}\) in place of the strongly connected component of \(2 - 2\). The minimum flow analysis is then performed on the acyclic graph Fig. \ref{fig:1}(b). This results in a minimum number of Test Paths as \(1\) – the path being \(1,2^{'},3\). Replacing \(2^{'}\), we get the path as \(1,2,2,3\) to correspond to Fig. \ref{fig:1}(a). However, the minimum number of paths to cover the Prime Paths is \(2\) - \(\{1,2,3\}\) and \(\{1,2,2,3\}\). Our method overcomes this inaccuracy by appropriately handling strongly connected components. Further, our solution of minimum Test Paths is computed in \(O\left(\left\vert{V}\right\vert\left\vert{E}\right\vert\right)\).

\begin{figure}
		\centering
		\includegraphics[height=1.0cm]{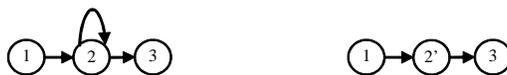}
		\caption{(a) Graph with a cycle (b) reduced to an acyclic graph.}
		\label{fig:1}
\end{figure}

The authors of \cite{Aho_Lee_1987} tailor a minimum flow algorithm specifically for minimum Test Paths for node coverage leading to the currently best known solution complexity of \begin{math}O\left(\left\vert{V}\right\vert\left\vert{E}\right\vert\right)\end{math}. Our algorithm in this paper also achieves this complexity. However, our technique is different and is based on the concept of Decreasing Paths \cite{Ciurea_Ciupala_2004}. 
The concept is similar to that of the Augmenting Path in the Ford-Fulkerson algorithm \cite{ford1962flows} to find the maximum flow. The Augmenting Path concept is well studied and numerous explanatory material is available. This allows our algorithm based on the Decreasing Path to be easily understood and conceptually simple. 
Further, our method is generic and handles different coverage criteria including node coverage.

Table \ref{table:contribution} summarizes the distinction of our method against related work.
\begin{table}
\caption{Comparison of our solution against related work}
\label{table:contribution}
\begin{tabular}{ | >{\centering\arraybackslash} m{0.12\textwidth} | >{\centering\arraybackslash} m{0.43\textwidth} | >{\centering\arraybackslash} m{0.40\textwidth} | }
	\hline
	Solution & Coverage Criteria & Comments \\
	\hline
	\cite{Ammann_Offutt_2008} & Prime Path & Does not attempt to minimize\\
	\hline
	\cite{Aho_Lee_1987} & Node Coverage & Minimization in \(O\left( \vert{V}\vert \vert{E}\vert\right)\) \\
	\hline
	\cite{Ntafos_Hakimi_1979} & Different structural coverage criteria & In-accuracy in solution. Minimization in \(O\left( {\vert{V}\vert}^2 \vert{E}\vert \right)\) \\
	\hline
	\cite{Li_Li_Offutt_2012} & Different structural coverage criteria & Heuristic solution (non-optimal) \\
	\hline
	Our Work & Different structural coverage criteria & Minimization in \(O\left( \vert{V}\vert \vert{E}\vert\right)\) \\
	\hline
\end{tabular}
\end{table}

Our method to obtain the minimum number of Test Paths is presented below.

\section{Generating the minimum number of Test Paths \label{section:3}}
Given a graph \(G_1=\left(V_{1},E_{1}\right)\), we need to find the minimum number of Test Paths that cover the set of test requirements \(TR\). We focus for the case of \(TR\) = set of Prime Paths of \(G_{1}\) and cover other coverage criteria in section \ref{section:OthCovCrit}. The minimum Test Paths are identified through a series of steps as shown below.
\begin{enumerate}
	\item First, the set of Prime Paths, \(P\) of \(G_{1}\) is computed. We present a method to obtain the lower bound on the minimum number of Test Paths.
	\item \(G_{1}\) is converted into a transform graph \(G_{2}=\left(V_{2},E_{2}\right)\), where \(V_{2}\) is the set of Prime Paths, \(P=\{p_{1},p_{2}, \ldots ,p_{n}\}\) and \( \left( p_{1},p_{2} \right) \in E_{2}\) if a path exists from \(p_{1}\) to \(p_{2}\) on \(G_{1}\) such that the path does not include any other Prime Paths (i.e. other than \(p_{1}\) and \(p_{2}\)).
	\item The transform graph \(G_{2}\) is made into an acyclic graph, \(G_{3}\). The work in \cite{Ntafos_Hakimi_1979} removed cycles from the original graph, \(G_{1}\), instead of the transform graph, \(G_{2}\). This inaccuracy in \cite{Ntafos_Hakimi_1979} can lead to incorrect results (refer Fig.~\ref{fig:1}).
	\item \(G_{3}\) is converted into a flow graph, \(G_{4}\), with new vertices and edges introduced. Every edge is annotated with lower bounds and capacities of flows. The minimum flow on \(G_{4}\) is computed through a two-step process. First, \(G_{4}\) is initialized with a feasible flow. We present a novel initialization algorithm in this paper. Second, the generic Decreasing Path algorithm \cite{Ciurea_Ciupala_2004} is used to find the minimum flow. Our initialization algorithm ensures the Decreasing Path algorithm can compute with a complexity of \begin{math}O\left(\left\vert{V}\right\vert\left\vert{E}\right\vert\right)\end{math}.
	\item The minimum flow in \(G_{4}\) is now interpreted as Test Paths in \(G_{1}\). 
\end{enumerate}
We detail each step of the process below. We will also use Fig.~\ref{fig:2} as a running example to help explain the algorithms used. 

\subsection{Generating the set of Prime Paths}
Algorithm \ref{algo:1} generates the set of test requirements \(TR\) (Prime Paths in this case). We use known methods (\cite{Ammann_Offutt_2008} and \cite{Ammann_Offutt_Web}) to compute the set of Prime Paths.

For the example graph in Fig. \ref{fig:2}, there are \(10\) Prime Paths: \(P=\{p_{0}=\{s,1,3,4,5\}, p_{1}=\{3,4,1,2,t\}, \allowbreak p_{2}=\{5,4,1,2,t\},\allowbreak p_{3}=\{1,3,4,1\},\allowbreak p_{4}=\{s,1,2,t\},\allowbreak p_{5}=\{3,4,1,3\},\allowbreak p_{6}=\{5,4,1,3\},\allowbreak p_{7}=\{4,1,3,4\},\allowbreak p_{8}=\{5,4,5\},\allowbreak p_{9}=\{4,5,4\}\}\).

\begin{algorithm}[H]
	\caption {Generating the set of Prime Paths}
	\label {algo:1}
	\KwIn{\(G_{1}=\left(V_{1},E_{1}\right)\), with \(\{s,t\} \in V_{1}\)}
	\KwOut{Set of Prime Paths, \(P=\{p_{1},p_{2}, \ldots ,p_{n}\}\)}
	initialize \(P^{'} = \{p_{1}, p_{2}, \ldots, p_{n} \} = \{e_{1}, e_{2}, \ldots, e_{n}\}, \forall e_{i} \in E_{1} \)\;
	\While{\(p_{i} \in P^{'} \forall i\) is not explored} 
	{
		\If{ \( p_{i}\) is not a cycle}
		{
			\If { \(p_{i}\) can be extended by edge \(e \in E_{1}\)}
			{
				\If {\(e\) does not visit a vertex in \(p_{i}\)} 
				{
					\( p_{i} += e\)\;
				}
			}
		}		
	}
	sort \(P^{'}\) in ascending order of size\;
	\For {\(i = \left(\vert{P^{'}}\vert\right)\) \KwTo \(1\) }
	{
		\If{ \(p_{i}\) is not a sub-path of any other path in \(P\)}
		{
			add \(p_{i}\) into \(P\)\;
		}
	}
	output \(P\)\;		
\end{algorithm}

We define four categories of Prime Paths. \(Type\,S\) are those Prime Paths that visit the vertex \(s\). \(Type\,T\) are those that visit the vertex \(t\). \(Type\,C\) are those Prime Paths that are cycles, and \(Type\,P\) are those that are simple paths and do not visit either \(s\) or \(t\). In the given example, we have the following cardinality - \(\vert{Type\,S}\vert = 2\), \(\vert{Type\,T}\vert = 3\), \(\vert{Type\,C}\vert = 5\), \(\vert{Type\,P}\vert = 1\). These categories hold for other test requirements (i.e. Nodes, Edges, Edge-pair, etc) as well.

\begin{lemma}
	\label{lemma:1}
	\(\max \left(\vert{Type \,S}\vert, \vert{Type \,T}\vert\right) \) is a lower bound on the minimum number of Test Paths for Prime Path coverage.
\end{lemma}
\begin{proof}
	Consider \(P_{T}=\{p_{1},p_{2}, \ldots,p_{n}\}\) to be a set of \(Type\,T\) Prime Paths. Since each of these Prime Paths visits vertex \(t\), the last node of every Prime Path \(p_{i} \in P_{T}\) is the vertex \(t\). Now, consider an \(s-t\) path that tours \(p_{1}\). This \(s-t\) path cannot tour any other Prime Path of \(P_{T}\) after touring \(p_{1}\) since the vertex \(t\) has an out-degree of \(0\). Thus, there cannot be any path in the graph where more than \(1\) Prime Paths of \(P_{T}\) can be toured simultaneously. Therefore, there would be at-least \(Type\,T\) \(s-t\) paths to cover all Type T Prime Paths. Similarly, there would be atleast \(Type\,S\) \(s-t\) paths to cover all \(Type\,S\) Prime Paths. Thus the number of s-t paths that will cover all Prime Paths will be atleast as much as \(\max \left(\vert{Type \,S}\vert, \vert{Type \,T}\vert\right) \). \qed
\end{proof}

The categorization of the Prime Paths can be achieved by checking the first and last vertex of every Prime Path. This can be done in \(O\left(\vert{P}\vert\right)\). For the example graph of Fig. \ref{fig:2}, the lower bound for the minimum number of Test Paths for Prime Path coverage \(= \vert Type\,T \vert = 3\).

\begin{figure}
	\begin{minipage}{0.5\linewidth}
		\centering
		\includegraphics[height=3.4cm]{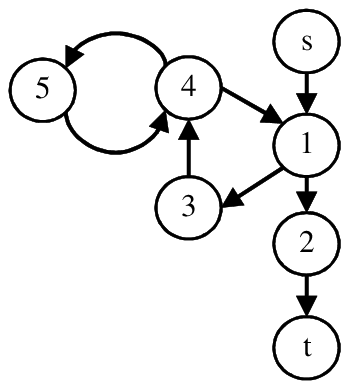}
		\caption{Example Input Graph, \(G_{1}\).}	
		\label{fig:2}
	\end{minipage}%
	\begin{minipage}{0.5\linewidth}
		\includegraphics[height=3.4cm]{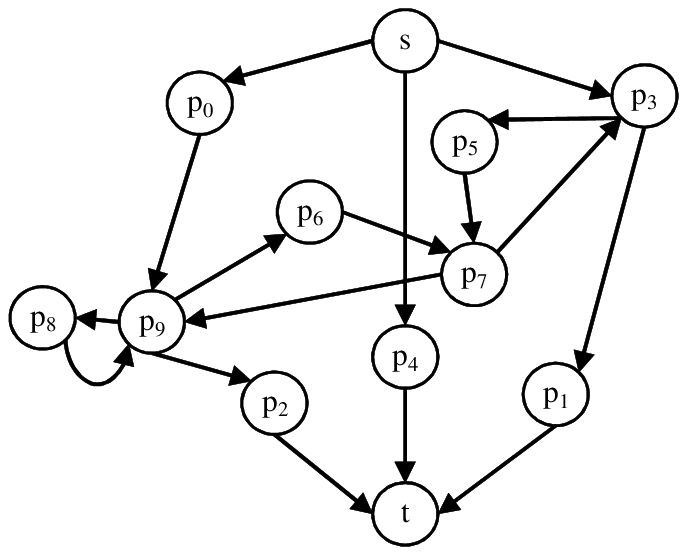}
		\caption{Transform Graph, \(G_{2}\) of \(G_{1}\).}
		\label{fig:3}
	\end{minipage}
\end{figure}

\subsection{Generating the Transform Graph} \label{section:TransGrph}
The Prime Paths are represented as vertices in a new graph, \(G_{2}\). Edges are placed between two vertices in \(G_{2}\) if a path in \(G_{1}\) can tour the two Prime Paths represented by the vertices. The problem of Prime Path coverage is thus transformed into a problem of node coverage (i.e. we want to identify \(s-t\) paths such that all vertices of \(G_{2}\) are covered). 
The algorithm for generation of the transform graph is provided below.

\begin{algorithm}[H]
	\caption {Generation of the Transform Graph}
	\label {algo:2}
	\KwIn{\(G_{1}, P\)}
	\KwOut{Transform Graph \(G_{2}=\left(V_{2}, E_{2}\right)\)}
	create new vertex \(s,t\) in \(V_{2}\); create new vertex \(v_{i}\) in \(V_{2}, \forall i \in P\) \;
	\ForEach{\(p_{i} \in \left(P + \{s\}\right)\)}
	{
		\ForEach{\(p_{j} \in \left(P - p_{i} + \{t\}\right)\)}
		{
			\(Path_{ij}=p_{i} \cup p_{j} \) \label{algo2:cup} \;
			\If{\(Path_{ij}\) contains only the prime paths \(p_{i}\) and \(p_{j}\) \label{algo2:ppCheck}}
			{
				add edge \(\left(p_{i},p_{j}\right)\) into \(E_{2}\)\;
			}
		}
	}
	output \(G_{2}=\left(V_{2},E_{2}\right)\)\;
\end{algorithm}

The Operator \(\cup\) in Step \ref{algo2:cup} works as follows. Consider the case of \(p_{i}=\{4,5,4\}\) and \(p_{j}=\{5,4,5\}\), then \(p_{i} \cup p_{j} = \{4, 5, 4, 5\} \) by observing that the paths \(p_{i}\) and \(p_{j}\) have overlapping vertices. In case of no overlap, \(p_{i} \cup p_{j} = \{p_{i} + \text{\{shortest path from the last node of }p_{i} \text{ to the first node of }p_{j} \text{\}} + p_{j}\}\).
Checking for the shortest path in \(G_{1}\) can be done using breadth first search with a complexity of \(O\left(E_{1}\right)\). Checking for the presence of a Prime Path in \(Path_{ij}\) (step \ref{algo2:ppCheck}) can be done in \(O\left(\vert{P}\vert\right)\). Thus, the transform graph can be computed in \(O\left({\vert{P}\vert}^{2} * \max\left(\vert{P}\vert,\vert{E_{1}}\vert\right)\right)\). The transform graph of Fig. \ref{fig:2} is shown in Fig. \ref{fig:3}.

An incomparable vertex set, \(I\), is one where for every pair of vertices \(v_{i},v_{j} \in I; v_{i}\) does not reach \(v_{j}\). From Fig. \ref{fig:3}, it can be visually seen that the three \(Type\,T\) Prime Paths, \(p_{1},p_{2},p_{4}\), form an incomparable vertex set, i.e.  \(I=\{p_{1},p_{2},p_{4}\}\). The maximum incomparable vertex set, \(I_{max}\) equals the minimum number of Test Paths through Dilworth's theorem for acyclic directed graphs and the proof in \cite{Ntafos_Hakimi_1979} for general graphs. From Lemma \ref{lemma:1}, \(\vert{I_{max}}\vert \geq \vert{I}\vert \).

\subsection{Removing Cycles}
The transform graph, \(G_{2}\), may have cycles. The minimum flow algorithm works over a directed acyclic graph and thus the cycles need to be removed. We achieve this by replacing a cycle with a new vertex. All incoming edges of any node in the cycle become incoming edges of the new vertex. Similarly, any outgoing edge of any vertex in the cycle become outgoing edges of the new vertex. Algorithm \ref{algo:3} presents the generation of the acyclic transform graph.
\begin{algorithm}[H]
	\caption {Reduction of Cycles in Transform Graph}
	\label {algo:3}
	\KwIn{\(G_{2}\)}
	\KwOut{Acyclic Transform Graph \(G_{3}=\left(V_{3},E_{3}\right)\)}
	initialize \(V_{3}=V_{2},E_{3}=E_{2},G_{3}=\left(V_{3},E_{3}\right)\)\;
	\While{\(G_{3}\) contains a cycle}
	{
		let \(p_{c}=\{v_{1},v_{2},\ldots,v_{n},v_{1}\}\) be a cycle\;
		record vertices \(v_{i} \in V_{3}\) which have an edge in \(E_{3}\) with vertices in \(p_{c}\)\; \label{algo3:recordInfo}
		remove vertices \(\{v_{1},v_{2},\ldots,v_{n}\}\) from \(V_{3}\); create new vertex \(v_{new}\) in \(V_{3}\)\;
		\ForEach{vertex \(v_{i} \in p_{c}\) \label{vertexReductionStart}}
		{
			\If {\(v_{i}\) has an edge from vertex \(v_{k}\) in \(G_{2}\) where \(v_{k} \in V_{3}-p_{c}\) }
			{
				remove edge \(\left(v_{k},v_{i}\right)\) from \(E_{3}\) \& create edge \(\left(v_{k},v_{new}\right)\) in \(E_{3}\)\;
			}
			\If{\(v_{i}\) has an edge to vertex \(v_{k}\) in \(G_{2}\) where \(v_{k} \in V_{3}-p_{c}\)}
			{
				remove edge \(\left(v_{i},v_{k}\right)\) from \(E_{3}\) \& create edge \(\left(v_{new},v_{k}\right)\) in \(E_{3}\)\;
			}
		  }\label{vertexReductionEnd} 
	}
	output \(G_{3}=\left(V_{3},E_{3}\right)\)\;
\end{algorithm}

Algorithm \ref{algo:3} works as follows. For the given transform graph, the cycles are computed using a trivial variant of Algorithm \ref{algo:1}. If the graph contains a cycle, any cycle is chosen and reduced to a new vertex. On the resulting graph, the cycles are identified again and the procedure is repeated. 
The complexity of steps \ref{vertexReductionStart} - \ref{vertexReductionEnd} is \(O\left(\vert{V_{2}}\vert\right)\) by assuming that the chosen cycle visits every vertex. Since, \(V_{2}\) represents the Prime Paths, this complexity can be written as \(O\left(\vert{P}\vert\right)\). If we assume every Prime Path is of \(Type\,C\) and reducing one cycle does not impact other Prime Paths, then the complexity of the entire algorithm can be said to be \(O\left({\vert{P}\vert * \max{\left(\vert{P}\vert, {P_{c}}\right)}}\right)\), where \(P_{c}\) is the complexity of finding a cycle.
After removing all the cycles, the acyclic transform graph is shown in Fig. \ref{fig:4}. 

\subsection{Generating the Minimum Flow}
The transform graph, \(G_{3}\), at this stage is a directed acyclic graph. We convert this graph into a flow graph. Each vertex \(v_{i} \in V_{3}-\{s,t\}\) is split into two vertices \(v_{i}^{+}\), \(v_{i}^{++}\). Let this new vertex set be represented as \(V_{4}\). New edges \(\left(v_{i}^{+},v_{i}^{++}\right)\) are also added. All incoming edges of \(v_{i}\) are made incoming edges of \(v_{i}^{+}\) and all outgoing edges of \(v_{i}\) are made outgoing edges of \(v_{i}^{++}\). Let this new edge set be represented as \(E_{4}\). Every edge \(\left(i,j\right) \in E_{4} \, \& \, i,j \in V_{4}\) is associated with a lower bound \(\left(l_{ij}\right)\) for a flow and an edge capacity \(\left(c_{ij}\right)\) as follows: 

\[
	l_{ij}= 
		\begin{cases}
			1 & \text{if} \left(i,j\right)= \left(v_{i}^{+},v_{i}^{++}\right) \\
			0 & \text{otherwise}
		\end{cases}
\]

\[
	c_{ij}= \frac{\vert{V_{4}}\vert}{2} - 1 		
\]

Let the resulting graph be represented as \(G_{4}=\left(V_{4},E_{4},L,C\right)\), where \(L\) is the set of lower bounds and \(C\) is the set of capacities. The introduction of new vertices, new edges and the flow requirements are pictorially depicted in Fig.\ref{fig:5}.

A feasible flow in the flow network is an assignment of a non-negative value, \(f_{ij}\) for every edge such that the following conditions hold. 
\begin{equation}
	\label{eq:1}
	f_{ij} \geq l_{ij} \,\&\, f_{ij} \leq c_{ij} , \forall \left(i,j\right) \in E_{4} \,\&\, i,j \in V_{4}
\end{equation}
\begin{equation}
	\label{eq:2}
	\sum_{i} f_{ij} = \sum_{j} f_{ji} ,\forall \left(i,j\right) \in E_{4} \,\&\, i,j \in V_{4}-\{s,t\}
\end{equation}
The flow of the network, \(f\), is defined as:
\begin{equation}
	\label{eq:3}
	f = \sum_j f_{sj} = \sum_i f_{it} , \forall \left(i,j\right) \in E_{4} \,\&\, i,j \in V_{4}
\end{equation}

The minimum flow,\(f_{min}\), is the least amount of feasible flow possible.

A flow in the network \(G_{4}\) can be mapped to an \(s-t\) path in \(G_{1}\). The flow conditions placed in \(G_{4}\)  ensures the edge \(\left(v_{i}^{+},v_{i}^{++}\right)\) is chosen in at-least one flow. Since this corresponds to every vertex \(v_{i} \in V_{3}-\{s,t\}\), the minimum flow will ensure every vertex of \(G_{3}\) is covered. By expanding vertexes placed in lieu of cycles in \(G_{3}\) , every vertex of \(G_{2}\) is covered. This ensures that every Prime Path of \(G_{1}\) is toured.

Every edge in \(G_{4}\)  is also annotated with a capacity, \(c_{ij}\), which represents the maximum flow that may be carried on the edge. The capacity is not directly interpretable from the Test Path perspective; however, the minimum flow algorithm requires the specification of a capacity. In the general case of computing the minimum flow in a network, the capacity determines if a feasible flow exists. In our case, we set the capacity in such a way that it guarantees a feasible flow.

\begin{theorem}
	\label{theorem:1}
	The capacity on an edge in \(G_{4}\) can be set as \(\frac{\vert{V_{4}}\vert}{2} - 1\). 
\end{theorem}
\begin{proof}
Consider the graph \(G_{3}\). The maximum number of \(s-t\) paths possible on \(G_{3}\) will be \(\left(\vert{V_{3}}\vert - 2 \right)\). This is because, the maximum will occur when every vertex of \(G_{3}-\{s,t\}\) leads to a new \(s-t\) path. Now, new vertices have been added in \(G_{4}\), but none of the new vertices will contribute to a new path over and above those that are possible in \(G_{3}\). Thus, the maximum capacity of every edge in \(G_{4}\) can be set as \(\left(\vert{V_{3}}\vert-2\right)\). But since \(\vert{V_{4}}\vert = 2 \ast \vert{V_{3}}\vert - 2 \Rightarrow \left(\vert {V_{3}} \vert -2 \right) =  \frac{\vert{V_{4}}\vert}{2} - 1\). \qed
\end{proof}

\begin{figure}
	\begin{minipage}{0.5\linewidth}
		\centering 
		\includegraphics[height=3.3cm]{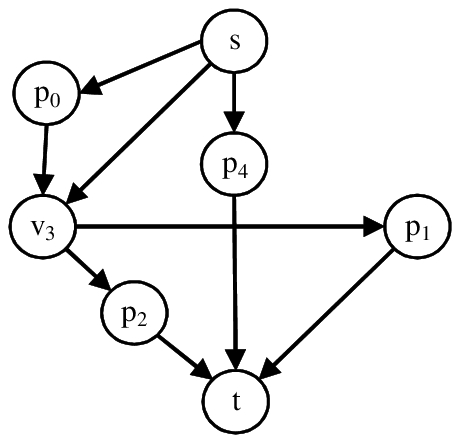}
		\caption{Creating an acyclic directed graph, \(G_{3}\) from \(G_{2}\).}
		\label{fig:4}
	\end{minipage}
	\begin{minipage}{0.5\linewidth}
		\centering 
		\includegraphics[height=3.3cm]{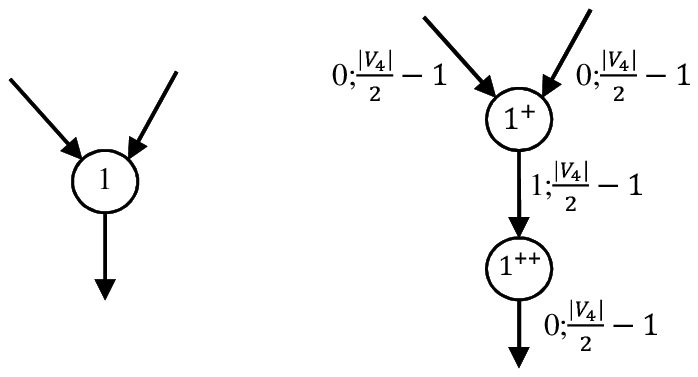}
		\caption{Vertex Splitting and flow requirements represented as lower bound; capacity}
		\label{fig:5}
	\end{minipage}
\end{figure}

Once the flow graph has been created, we determine the minimum flow through the technique similar to \cite{Ciurea_Ciupala_2004}. 

The minimum flow algorithm works in two steps, 
\begin{inparaenum} [\itshape a\upshape)]
\item an initial flow is placed on the flow graph such that all flow requirements (\ref{eq:1} \& \ref{eq:2}) are met; and 
\item the minimum flow is computed using the decreasing path algorithm. Using the algorithms for a) and b) as detailed in \cite{Ciurea_Ciupala_2004} over the network \(G_{4}\) provides a method with a complexity of \(O\left(\left\vert{V_{4}}\right\vert\left\vert{E_{4}}\right\vert c_{max}\right)\), where \(c_{max}\) represents the maximum value of the capacity, \(C\). From Theorem \ref{theorem:1}, \(c_{max}=O\left(\vert V_{4} \vert\right)\). Therefore, the complexity of the method in \cite{Ciurea_Ciupala_2004} becomes \(O\left({\left\vert{V_{4}}\right\vert}^2 \left\vert{E_{4}}\right\vert\right)\). This complexity can be generalized as \(O\left({\left\vert{V}\right\vert}^2 \left\vert{E}\right\vert\right)\).  
\end{inparaenum}
Our work in this paper introduces a novel initialization algorithm (i.e. aspect a)) such that the complexity of computing the minimum flow using the standard decreasing path algorithm becomes \(O\left(\left\vert{V_{4}}\right\vert\left\vert{E_{4}}\right\vert\right)\). 
This complexity can be generalized to \(O\left(\left\vert{V}\right\vert\left\vert{E}\right\vert\right)\). 

\subsubsection{Initialization with a feasible flow}
Given the graph \(G_{4}=\left(V_{4},E_{4},L,C\right)\), Algorithm \ref{algo:4} places an initial flow such that the flow requirements \ref{eq:1} \& \ref{eq:2} are met. 

\begin{algorithm}[H]
	\caption {Initialization with a feasible flow}
	\label {algo:4}
	\KwIn{\(G_{4}=\left(V_{4},E_{4},L,C\right)\)}
	\KwOut{feasible flow \( f_{ij} , \forall \left(i,j\right) \in E_{4} \,\&\, i,j \in V_{4}\)}
	initialize \(f_{ij}=0,\forall \left(i,j\right) \in E_{4} \,\&\, i,j \in V_{4}\)\;
	\ForEach{ vertex \(i \in V_{4}-\{s,t\}\) }
	{
		find path, \(p_{s}\), from \(s\) to \(i\) using breadth-first-search\;
		find path, \(p_{t}\), from \(i\) to \(t\) using breadth-first-search\;
		path, \(p=p_{s}+ p_{t}\)\;
		\(k=\min⁡ \{\left(f_{mn} - l_{mn}\right) ,\forall \left(m,n\right) \in p\}\)\;
		\If {\(k < 0\)}
		{
			\ForEach{ edge \(\left(m,n\right) \in p\)}
			{
				\(f_{mn}+= 1\)\;
			}
		}
	}
\end{algorithm}

\begin{theorem}
	\label{theorem:2}
	Algorithm \ref{algo:4} ensures flow requirements are met.
\end{theorem}
\begin{proof}
Consider a vertex, \(i\), in \(G_{3}\). This vertex will now be represented by \(i^{+}\) and \(i^{++}\)in \(G_{4}\). A path from \(s\) to \(i^{++}\) will cover the edge \(\left(i^{+},i^{++}\right)\) since \(i^{++}\) is reachable only through \(i^{+}\). Thus, incrementing the flow along the path from \(s\) to \(i^{++}\) will ensure that the flow condition of \(l_{ij}=1\) for \(\left(i,j\right)= \left(i^{+},i^{++}\right)\) is met. By checking \(\forall i \in V_{4}\), ensures flow requirement (\ref{eq:1}) for all edges is met. The increments of the flow is done for every edge of a path from \(s\) to \(t\). Consider the vertex \(i^{+}\) in \(G_{4}\). Let it have \(m\) incoming edges. Since \(i^{+}\) is created by splitting vertex \(i\), \(i^{+}\) will have 1 outgoing edge. Let vertex \(i^{+}\) be part of an \(s-t\) path \(n\) number of times where \(1 \leq n \leq \frac{\vert{V_{4}}\vert}{2}-1\). The \(m\) incoming edges will be visited \(n\) times with each visit incrementing the flow by 1. Thus the sum  of flows on the incoming edges will be \(n\). Similarly, the outgoing edge will be visited \(n\) times and will also have a flow of \(n\). Thus, by incrementing the flow of every edge of an \(s-t\) path ensures flow requirement (\ref{eq:2}) is met. \qed
\end{proof}

Breadth first search has a complexity of \(O(\vert E\vert)\). Therefore Algorithm \ref{algo:4} has a complexity of \(O\left( \vert{E_{4}}\vert \vert{V_{4}}\vert\right)\).

Once an initial flow is placed on the flow graph, the minimum flow is computed using the decreasing path algorithm.

\subsubsection{Decreasing Path algorithm to find the minimum flow}
We use the generic Decreasing Path algorithm as detailed in \cite{Ciurea_Ciupala_2004}. The algorithm is based on the Augmenting Path concept of the Ford-Fulkerson algorithm \cite{ford1962flows}. The Decreasing Path concept can be stated as follows. Every edge in \(G_{4}\) is termed as a forward edge. Let this set be called \(E_{4}^{f}\).  For every forward edge, a new backward edge is introduced; i.e. a forward edge of the form \(\left(i,j\right)\), will have a backward edge of the form \(\left(j,i\right)\). Let this set of backward edges be called \(E_{4}^{b}\). For each backward edge, the lower bound, \(l_{ij}\), is set to \(0\) and the capacity, \(c_{ij}\) is set to \(\frac{\vert{V_{4}}\vert}{2}-1\). The residual capacity of an edge, \(r_{ij}, \forall \left(i,j\right) \in E_{4}^{f} \cup E_{4}^{b}\), is defined as follows:
\[
	r_{ij}= 
		\begin{cases}
			f_{ij} - l_{ij} & \text{if} \left(i,j\right) \in E_{4}^{f} \\
			c_{ij}- f_{ij} & \text{if} \left(i,j\right) \in E_{4}^{b}
		\end{cases}
\]

A decreasing path is a path from \(s\) to \(t\) where the residual capacity of every edge is greater than \(0\). If a decreasing path visits a forward edge, then the flow on the edge can be reduced. If on the other hand, the decreasing path visits a backward edge, then the flow on the corresponding forward edge has to be increased. The flow in the graph will be minimum when no more decreasing paths can be found. The minimum flow is the optimal solution as proved in \cite{Ciurea_Ciupala_2004}.

The algorithm to compute the minimum flow is presented as Algorithm \ref{algo:5}.

The complexity of the Decreasing Path algorithm can be determined using the technique in \cite{Ciurea_Ciupala_2004}. Each reduction of the flow in the network will need \(O\left(\vert{E_{4}}\vert\right)\) because of the path found through breadth first search. Since the maximum flow that may be initialized is \(\frac{\vert{V_{4}}\vert}{2}-1 \) (from Theorem \ref{theorem:1}), the minimum flow will be found in \(O\left(\vert{V_{4}}\vert \vert{E_{4}}\vert\right)\). This can be generalized to \(O\left( \vert{V}\vert \vert{E}\vert \right)\).

\begin{algorithm}[H]
	\caption {Decreasing Path Algorithm}
	\label {algo:5}
	\KwIn{\(G_{4}=\left(V_{4},E_{4},L,C\right)\), with initial flow}
	\KwOut{minimum flow \(f\) and flows on \(G_{4}\)}
	\ForEach{edge \(\left(i,j\right) \in E_{4}\)}
	{
		put \(\left(i,j\right)\) in \(E_{4}^{f}\)\;
		put \(\left(j,i\right)\) in \(E_{4}^{b};l_{ji}=0; c_{ji}= \frac{\vert{V_{4}}\vert}{2}-1 \)\;
	}
		\While{path, \(p\), exists from \(s\) to \(t\) using breadth-first-search such that \(r_{mn}>0,\forall \left(m,n\right) \in p\)}
		{
			\( r_{min} = \min \{ r_{mn} , \forall \left(m,n\right) \in p \}\)\;
			\ForEach {edge \(\left(m,n\right) \in p \)}
			{
				\eIf{ \(\left(m,n\right) \in E_{4}^{f}\)}
				{
					\(f_{mn} -= r_{min}\)\;
				}
				{
					\(f_{mn} += r_{min}\)\;
				}
					\(f_{nm} = f_{mn}\)\;
			}
		}
	output \(f=\sum_{j} f_{sj}, G_{4}\)\;
\end{algorithm}

The overall complexity of the minimum flow algorithm is the maximum of the complexity between the initialization and the decreasing path algorithm.

\begin{multline}
\max⁡\left(\text{initialization algorithm,decreasing path algorithm}\right)= \\ 
\max \left( O\left( \vert{V}\vert \vert{E}\vert \right),O\left( \vert{V}\vert \vert{E}\vert \right) \right) = O\left( \vert{V}\vert \vert{E} \vert \right)
\end{multline}

Thus, the complexity of our method of \(O\left( \vert{V}\vert \vert{E}\vert \right)\) improves upon the complexity of the generic minimum flow algorithm of \(O\left( \vert{V}\vert^{2} \vert{E}\vert \right)\).

For the running example, the minimum flow is computed as \(3\) which can be quite clearly observed from Fig. \ref{fig:4}. The minimum flow incidentally equals the lower bound for the example.

\subsection{Generating the Minimum Test Paths from the Minimum Flow} \label{section:minTestPaths}
The minimum number of \(s-t\) paths to cover all the Prime Paths is the flow, \(f\), in \(G_{4}\). The \(s-t\) paths in \(G_{4}\) and the corresponding paths in \(G_{1}\) are identified using Algorithm \ref{algo:6}.

Steps \ref{algo6:startPathG3} to \ref{algo6:endPathG3} identify the paths corresponding to the minimum flow in \(G_{3}\). For the running example, the \(path^{G3}=\{s,p_{4},t\},\{s,v_{3},p_{1},t\}\,\text{and}\, \{s,p_{0},v_{3},p_{2},t\}\). 

Vertices in \(path^{G3}\) which represent cycles in \(G_{2}\) are replaced. In the running example, there are no cycles to be introduced in the first path \(\{s,p_{4},t\}\). Replacing \(v_{3}\) in the second path gives \(\{s,p_{3},p_{5},v_{2},p_{3},p_{1},t\}\). Note that although \(v{3}\) represents the cycle - \(\{v_{2}, p_{3}, p_{5}, v_{2}\}\), it is replaced as \(\{p_{3}, p_{5}, v_{2},p_{3}\}\) such that an edge exists between \(s\) and \(p_{3}\). Step \ref{algo6:correctCycle} performs this operation. Replacing  \(v_{2}\) gives \(\{s,p_{3},p_{5},p_{7},v_{1},p_{6},p_{7},p_{3},p_{1},t\}\). Replacing \(v_{1}\) gives the path \(\{s,p_{3},p_{5}, \allowbreak p_{7},p_{9},p_{8}, \allowbreak  p_{9},p_{6},p_{7}, \allowbreak  p_{3},p_{1},t\}\). Similarly, after replacing the cycles in the third path, we get \(\{s,p_{0},p_{9},p_{8},\allowbreak p_{9},p_{6},p_{7}, \allowbreak p_{9},p_{8},\boldsymbol{p_{9}, \allowbreak p_{3}},p_{5},p_{7},\allowbreak  p_{9},p_{8},\allowbreak p_{9},p_{6},\boldsymbol{p_{7},\allowbreak p_{2}},t\}\). 

\begin{algorithm}[H]
	\caption {Identifying minimum Test Paths from minimum flow}
	\label {algo:6}
	\KwIn{\(G_{4}=\left(V_{4},E_{4}, L, C\right)\), with flow on each edge \& minimum flow \(f\)}
	\KwOut{All s-t paths, \(path^{G1}\), on \(G_{1}\) corresponding to the minimum flow}
	Remove all backward edges of \(G_{4}\). Merge vertices of the form \(\{v^{+},v^{++}\}\) into \(v\). Incoming edges of \(v\) = incoming edges of \(v^{+}\). Outgoing edges of \(v\) = outgoing edges of \(v^{++}\). Let the resulting graph be \(G_{3}\)\;
	\For {\(i = 0\) \KwTo \(f\)  }
	{
		remove all edges from \(G_{3}\) which have flow of \(0\)\; \label{algo6:startPathG3}
		find path, \(path^{G3}\), from \(s\) to \(t\) using breadth-first-search\;
		\ForEach{ edge\(\left(m,n\right) \in path^{G3}\) }
		{
			\(f_{mn} -= 1\) \;
		} \label{algo6:endPathG3}
		\While {\(path^{G3}\) contains a vertex, \(v^{G3}_{c}\), reduced from a cycle, \(c\) \label{algo6:startReduceCycles}}
		{
			Let the cycle be represented as \(c=\{v_{1}^{G2},v_{2}^{G2}, \ldots, v_{1}^{G2}\}\)\;
			replace \(v^{G3}_{c}\) with \(c=\{v_{j}^{G2},v_{j+1}^{G2}, \ldots, v_{j}^{G2}\}\) where \(v_{c-1}^{G3}\) is connected to \(v_{j}^{G2}\) through the information recorded in Algorithm \ref{algo:3} Step \ref{algo3:recordInfo}\; \label{algo6:correctCycle}
		} \label{algo6:endReduceCycles}
		\(path^{G2} = path^{G3} \)\;
		ensure an edge exists between every vertex of \(path^{G2}\) in \(G_{2}\)\; \label{algo6:ensureConnectivity}
	} 
	\ForEach {\(path^{G2}\) }
	{
		\ForEach {cycle, \(c^{'}\) in \(path^{G2}\) \label{algo6:startRemoveRedundancy} }
		{
			\If {number of occurrences of \(c^{'}\) in all paths \(path^{G2}>1\)}
			{
				except the first instance, replace all other instances of \(c^{'}\) in \(path^{G2}\) with the first vertex of \(c^{'}\)\; \label{algo6:removeMultiCycles}
			}
			let \(path_{sub}=c^{'}-\)first \& last vertices of \(c^{'}\)\;
			\If {number of occurrences of \(path_{sub}\) in all paths \(path^{G2}>1\)}
			{
				replace \(c^{'}\) in \(path^{G2}\) with the first vertex of \(c^{'}\)\; \label{algo6:removeMultiCyclesSub-Paths}
			}
		} \label{algo6:endRemoveRedundancy}
		initialize \(path^{G1}=s\)\; \label{algo6:startPathG1}
		\ForEach {vertex, \(v_{i}^{G2} \in path^{G2}\)}
		{
			\(path^{G1}= path^{G1} \cup  \text{prime path represented by} \,v_{i}^{G2}\)\;
		}
		\(path^{G1}= path^{G1} \cup t\) \; \label{algo6:endPathG1}
		Output \(path^{G1}\)\;
	} 
\end{algorithm}

Note that the vertices \(p_{9}\) and \(p_{3}\) in the sub-path \(\{p_{9},p_{3}\}\) of the third Test Path are not directly connected in \(G_{2}\). Similarly, the vertices \(p_{7}\) and \(p_{2}\) are not directly connected. This aspect of connectedness is taken care in step \ref{algo6:ensureConnectivity}.

As can be seen from the paths, redundancy exists. For example, the cycle \(\{p_{9},p_{8},p_{9}\}\) is present 3 times in the third path. We remove this redundancy by replacing such cycles with the first node of the cycle (\(p_{9}\) in this case) for all occurrences except the first (step \ref{algo6:removeMultiCycles}).
Also note that for the cycle \(\{p_{7},p_{9},p_{6},p_{7}\}\) the sub-path \(\{p_{9},p_{6}\}\) is occurring more than once. This implies the cycle can be again reduced to the first node of the cycle (step \ref{algo6:removeMultiCyclesSub-Paths}). Thus, the set of paths after removing redundancy is \(\{s,p_{4},t\},\{s,p_{3},p_{5},p_{7},p_{9},p_{8},p_{9},p_{6},p_{7},p_{3},p_{1},t\},\{s,p_{0},p_{9},p_{2},t\}\). 

The last step of the algorithm is to merge the Prime Paths to obtain the \(s-t\) paths on \(G_{1}\). This operation is performed by the operator \(\cup\) and is as explained in Section \ref{section:TransGrph}. 


Thus, the minimum \(s-t\) paths for the running example needed to cover all Prime Paths are \(\{s,1,2,t\}\),\(\{s,1,3,\allowbreak 4,1,3,4,5,\allowbreak 4,5,4,1,\allowbreak 3,4,1,2,t\}\) and \(\{s,1,3\allowbreak ,4,5,4, \allowbreak 1,2,t\}\). Observe that the Test Paths are long in length. The length of the Test Paths are the number of edges and equals \(27\) in this example.  

\subsection{Other Coverage Criteria \label{section:OthCovCrit}}
Our method of identifying the minimum number of Test Paths is generic and can cater to all of the structural coverage criteria of \cite{Ammann_Offutt_2008} as shown below.

\subsubsection{Edge-Pair Coverage Criterion}
For the Edge-Pair coverage criterion, the test requirement set, \(TR\), is the set of all paths of length at-most \(2\). An edge-pair can be represented as a path \(\{v_{i}, v_{j}, v_{k}\}\), where \(\left(v_{i}, v_{j}\right) \& \left(v_{j}, v_{k}\right) \in E\). The minimum number of Test Paths such that \(TR\) is covered can be obtained directly from our method. Algorithm \ref{algo:1} will have to be modified to generate the set of edge-pairs. Other algorithms can be used exactly as presented. Further, Lemma \ref{lemma:1} and the Theorems hold. From Lemma \ref{lemma:1}, the lower bound on the minimum number of Test Paths would be \(\max\left(\vert{Type S}\vert,\vert{Type T}\vert \right) \), where \(Type \, S\) and \(Type\,T\) are those test requirements that contain the vertex \(s\) and \(t\) respectively. 

\subsubsection{Simple and Complete Round Trip Coverage Criterion}
A Round Trip path is a Prime Path of \(Type \, C\). The test requirement of the Simple Round Trip coverage criterion contains at least one \(Type \, C\) Prime Path which begins and ends for a given vertex. The test requirement for Complete Round Trip coverage criterion contains all \(Type \, C\) Prime Paths for a given vertex. Therefore, the Round Trip coverage criteria focuses on a subset of the set of Prime Paths for a given graph. In our formulation, algorithm \ref{algo:1} can be suitably modified to choose the set of Prime Paths needed for the set \(TR\). The minimum number of Test Paths needed for this set of \(TR\) is directly obtained by the other algorithms. However, in this case of Round Trip coverage, Lemma \ref{lemma:1} will provide a value of \(0\) as the lower bound since there are no \(Type \, S\) or \(Type\,T\) test requirements. 

\subsubsection{Edge Coverage Criterion}
To handle Edge Coverage, there would be no need for algorithm \ref{algo:1} as the test requirements are directly available as \(E\). The other algorithms can be used exactly as presented. Algorithm \ref{algo:2} represents every edge as a vertex in the flow network. Algorithm \ref{algo:3} reduces the graph to an acyclic one. The minimum flow computation using algorithms \ref{algo:4} and \ref{algo:5} with algorithm \ref{algo:6} will give the minimum number of Test Paths needed to cover every edge. Lemma \ref{lemma:1} will equal the maximum of number of test requirements that contain \(s\) or \(t\) which in this case is the maximum of the out-degree of \(s\) and the in-degree of \(t\).

\subsubsection{Node Coverage Criterion}
Node Coverage can be handled by using a sub-set of the algorithms in our method. We directly use algorithms \ref{algo:3} to \ref{algo:6} to obtain the minimum number of Test Paths for node coverage. Lemma \ref{lemma:1} will equal the maximum of the number of test requirements that contain \(s\) or \(t\) which in this case will equal \(1\).

\section{Experimental Results \label{section:4}}

We have evaluated our method through two sets of tests. 
In the first test we use 18 graphs representing actual open-source software~\footnote{We thank Nan Li of \cite{Li_Li_Offutt_2012} for sharing the manually created graphs.} as test inputs. Since these graphs were also used as test inputs in \cite{Li_Li_Offutt_2012}, our results can be directly compared with that of \cite{Li_Li_Offutt_2012}. Tables \ref{table:expResults1} \& \ref{table:expResults2} provides the results of the first test.

The results show a significant reduction in the number of Test Paths generated over the methods of \cite{Ammann_Offutt_2008} and \cite{Li_Li_Offutt_2012}. Averaging over the 18 graphs, the number of Test Paths were reduced by 72.9\% over the solution of \cite{Ammann_Offutt_2008}. The reduction varied from a minimum of 16.3\% to a maximum of 89.9\%. Comparing with \cite{Li_Li_Offutt_2012}, the number of Test Paths from our solution reduced by 59.6\% on average and varied from a minimum benefit of 0\% to a maximum benefit of 91.4\%. 

Although our solution does not explicitly attempt to minimize the length of the Test Paths, the results show that in many cases, the Test Path length is reduced as well. The average reduction in the Test Path length on the 18 graphs was -9.4\% over \cite{Ammann_Offutt_2008} (i.e. on average, the test path length increased) and 38.5\% over \cite{Li_Li_Offutt_2012}. Note that in some cases, the minimization of the number of Test Paths has actually increased the length of the Test Paths.

The identification of the lower bound (Lemma \ref{lemma:1}) has performed exceedingly well and on average over the 18 graphs, was only 8.2\% outside the true value. 

The execution time of our method is significant with Algorithm 2 being the most costly. On average, the execution time of our method was 338 times that of \cite{Ammann_Offutt_2008}. In absolute terms, the average execution time was 3.6 minutes, but varied from a minimum of 5 ms to a maximum of 27.6 minutes.

\begin{table} 
\caption{Test Paths \& Execution Time (in ms) over graphs representing actual software}
\label{table:expResults1}
\begin{minipage}{\textwidth}
\begin{tabular}{ p{0.45\textwidth}| c | c | c | c | c | c | c | c | c |c | c}
	\hline
	Num. of Prime Paths & 9 &	11 &	27	&	27	&	35	&	38	&	46	&	63	&	69 & 78~\footnote{maps to 71 in results (Table II) of \cite{Li_Li_Offutt_2012} as our solution uses a single source and sink.} & 93~\footnote{similarly maps to 84 in results (Table II) of \cite{Li_Li_Offutt_2012}.} \\
	\hline
	Num. of Test Paths from Solution in \cite{Ammann_Offutt_2008} & 7 & 9 & 14 & 19 & 26 & 22 & 36 & 34 & 49 & 37 & 69\\
	Test Paths' Length from Solution in \cite{Ammann_Offutt_2008} & 46 & 56 & 200 & 181 & 230 & 210 & 426 & 506 & 544 & 675 & 771\\
	Num. of Test Paths from our Solution  & 5	& 7 & 11	& 14	& 20	& 12	& 17	& 11	& 41 & 18 & 54\\
	Lower Bound on Num. of Test Paths & 5	& 7	& 11	& 14	& 20	& 12	& 8	& 11	& 35 & 18 & 54\\
	Test Paths' Length from our Solution  & 33	 & 42	 & 159	 & 133	 & 223	 & 154	 & 294	 & 385	 & 527 & 764 & 1043\\
  \hline
Execution Time of Solution in \cite{Ammann_Offutt_2008} & 2 & 2 & 10 & 2 & 7 & 2 & 5 & 14 & 11 & 6 & 15\\
Execution Time of Algorithm 1 & 1 & 1 & 9 & 1 & 6 & 1 & 2 & 12 & 6 & 5 & 12\\
Execution Time of Algorithm 2 & 1 & 1 & 21 & 7 & 36 & 13 & 50 & 153 & 63 & 170 & 150\\
Execution Time of Algorithm 3 & 1 & 1 & 3 & 2 & 5 & 1 & 8 & 19 & 4 & 13 & 22\\
Execution Time of Algorithm 4 \& 5 & 1 & 1 & 2 & 1 & 6 & 1 & 2 & 5 & 6 & 1 & 4\\
Execution Time of Algorithm 6 & 1 & 1 & 11 & 4 & 26 & 7 & 21 & 48 & 23 & 38 & 40\\
Total Execution Time of our Solution & 5 & 5 & 46 & 15 & 79 & 23 & 83 & 237 & 102 & 227 & 228\\
\hline
\end{tabular}
\end{minipage}
\end{table}

In the second test, we created 4391 random graphs of varying complexity as test inputs. The Prime Paths of these graphs varied from 7 to 150. The number of Test Paths, the length of Test Paths and the execution time were recorded averaging over graphs of a particular Prime Path. The execution times were the average of 5 runs over the same graph. The results are shown in Fig. 6. Similar to the first test, we see a significant improvement in terms of number of Test Paths. On average, our method reduced the number of Test Paths by 70.9\%. The length of Test Paths were reduced by 2.7\% on average. The lower bound was away from the true value by only 0.6\% on average. On average, our method took 0.39 seconds to execute which was 105 times the execution time of \cite{Ammann_Offutt_2008}.

Overall, our method of minimizing the number of Test Paths provides good results. To mitigate the concern of the increase in execution time, the quality of the lower bound can be exploited (for example attempting the minimization only when the solution from \cite{Ammann_Offutt_2008} is over 5 times the lower bound).
\begin{table}
\caption{Test Paths \& Execution Time (in ms) over graphs representing actual software}
\label{table:expResults2}
\begin{minipage}{\textwidth}
\begin{tabular}{ l| c | c | c | c | c | c | c }
	\hline
	Num. of Prime Paths & 98 & 101 & 122 & 170 & 1074~\footnote{similarly maps to 1024 in results (Table II) of \cite{Li_Li_Offutt_2012}.} & 1141 & 1844\\
	\hline
	Num. of Test Paths from method in \cite{Ammann_Offutt_2008}& 67 & 60 & 80 & 102 & 933 & 954 & 885\\
	Test Paths' length from method in \cite{Ammann_Offutt_2008}& 1096 & 872 & 1577 & 1986 & 27457 & 37245 & 19828\\
	Num. of Test Paths from our Method & 41 & 47 & 20 & 42 & 362 & 96 & 102 \\
	Lower Bound on Num. of Test Paths & 31 & 47 & 20 & 28 & 362 & 65 & 102\\
	Test Paths' Length from our Method & 967 & 1152 & 1541 & 1789 & 34733 & 37358 & 21482 \\
  \hline
Execution Time of Solution in \cite{Ammann_Offutt_2008} & 11 & 11 & 14 & 48 & 3548 & 5884 & 1997 \\
Execution Time of Algorithm 1 & 9 & 9 & 11 & 26 & 2376 & 4454 & 748\\
Execution Time of Algorithm 2 & 306 & 253 & 1057 & 1386 & 634719 &1240553& 1315738\\
Execution Time of Algorithm 3 & 42 & 13 & 79 & 222 & 72520 & 167950 & 95775\\
Execution Time of Algorithm 4 \& 5 & 6 & 4 & 1 & 12 & 137 & 19 & 87\\
Execution Time of Algorithm 6 & 58 & 71 & 230 & 179 & 42628 & 106613 & 242745\\
Total Execution Time for our Method & 421 & 350 & 1378 & 1825 & 752380 & 1519589 & 1655093\\
\hline
\end{tabular}
\end{minipage}
\end{table}

\begin{figure}
		\centering
	\begin{minipage}{0.6\linewidth}
		\includegraphics[height=3.5cm]{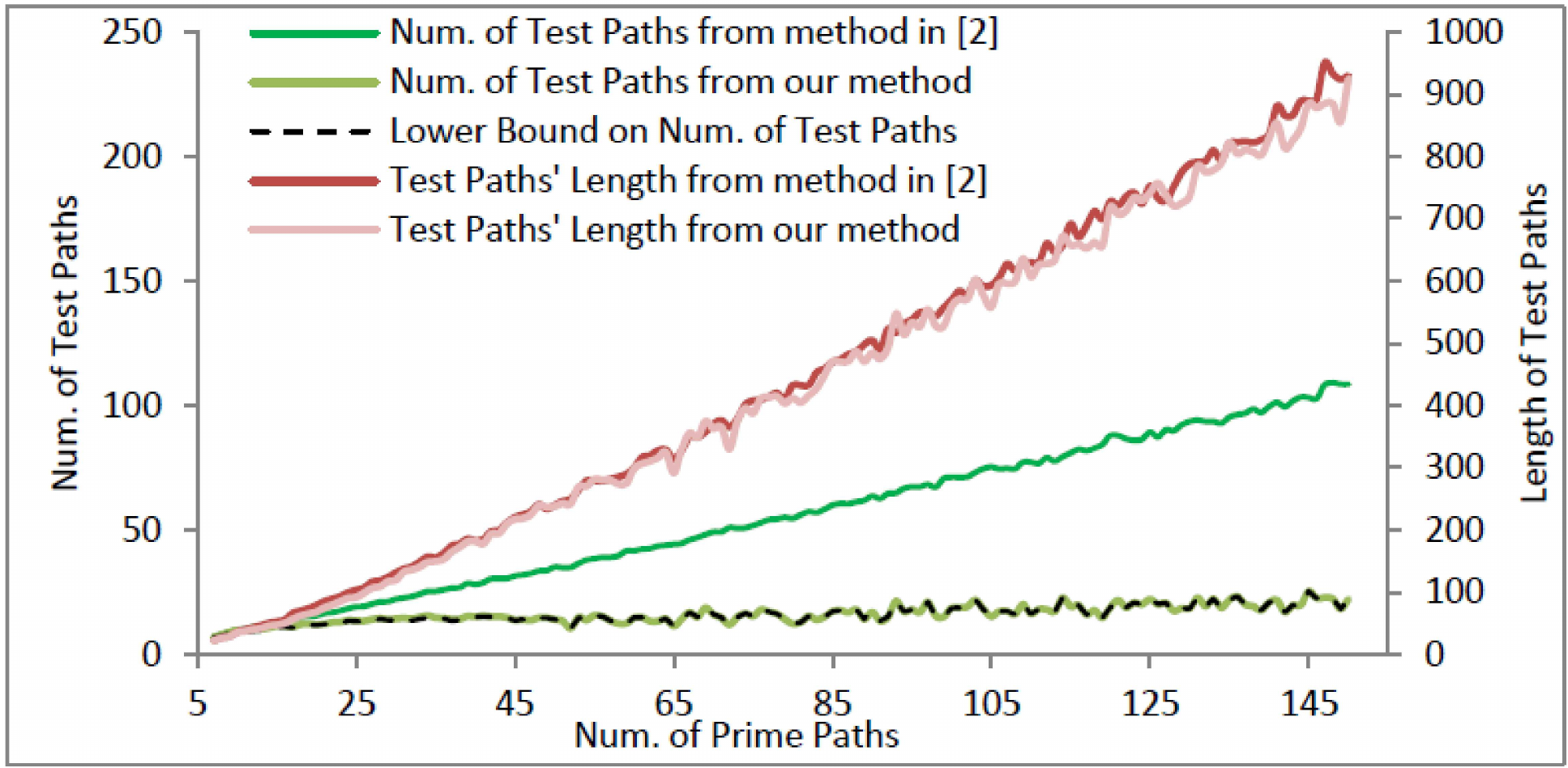}
	\end{minipage}%
	\begin{minipage}{0.35\linewidth}
		\includegraphics[height=3.5cm]{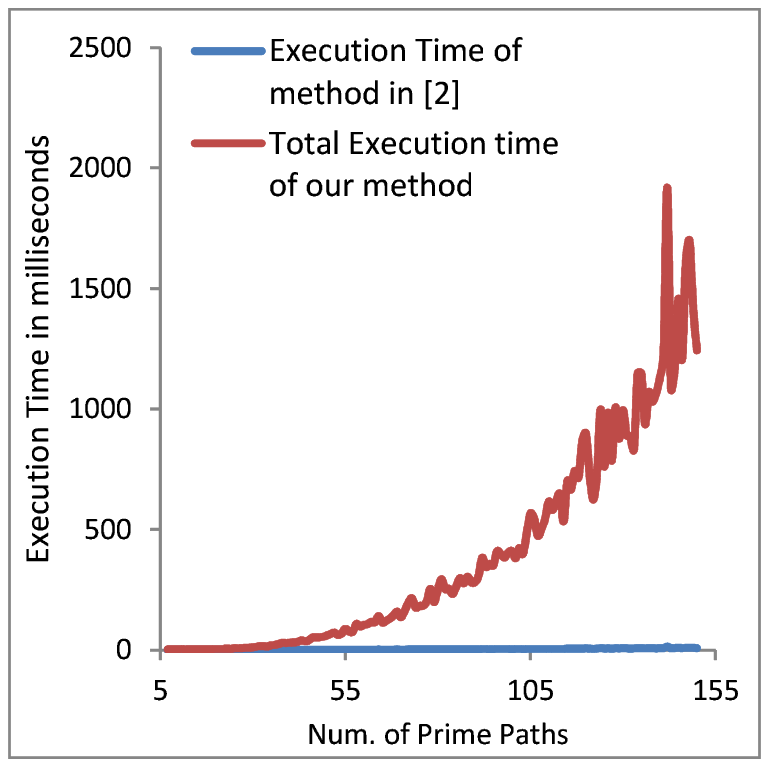}
	\end{minipage}
	\caption{Test Paths \& Execution Time (in ms) over randomly created graphs}	
	\label{fig:graph}
\end{figure}

\section{Conclusion \label{section:5}}
In this paper, we present a method to achieve the minimum number of Test Paths for 
Prime Path and other structural coverage criteria. The Prime Path criterion improves the quality of the Test Cases. However, even a small graph can have many Prime Paths and having a less number of Test Paths directly results in time saving over the Test Case execution. Our solution obtains an optimal solution with a time complexity of \(O\left( \vert{V}\vert\vert{E}\vert\right)\). This matches the best known time complexity till date. We have also presented a lower bound for the minimum number of Test Paths. Experimental results on graphs representing actual software and random graphs shows the superior performance of our method in terms of the number of Test Paths and the quality of the lower bound.
%
%
\bibliographystyle{splncs03}
\bibliography{bibliography_3}

\end{document}